\documentclass[conference]{IEEEtran}
\IEEEoverridecommandlockouts

\usepackage{cite}
\usepackage{amsmath,amssymb,amsfonts}
\usepackage{amsthm}
\usepackage{graphicx}
\usepackage{textcomp}
\usepackage{xcolor}
\usepackage{enumitem}

\usepackage[caption=false,font=footnotesize]{subfig}

\usepackage{booktabs}
\usepackage{algorithm}
\usepackage{algorithmic}
\usepackage{etoolbox}
\usepackage{bm}

\theoremstyle{plain}
\newtheorem{theorem}{Theorem}

\newtheorem{proposition}{Proposition}

\theoremstyle{definition}
\newtheorem{definition}{Definition}

\theoremstyle{remark}
\newtheorem{remark}{Remark}

\def\BibTeX{{\rm B\kern-.05em{\sc i\kern-.025em b}\kern-.08em
    T\kern-.1667em\lower.7ex\hbox{E}\kern-.125emX}}

\begin{document}

\title{Design of Outage-Limit-Approaching Protograph LDPC Codes via Generalized Rootchecks
\thanks{
This work was partly supported by Institute of Information \& communications Technology Planning \& Evaluation (IITP) grant funded by the Korea government (MSIT) (RS-2024-00398449, RS-2024-00397216) and Basic Science Research Program through the National Research Foundation of Korea (NRF) funded by the Ministry of Science and ICT (RS-2024-00343913) and the Ministry of Education (RS-2023-00247197). }
} 

\author{
Inki Kim$^{1}$,
Hyuntae Ahn$^{1}$,
Yongjune Kim$^{3}$,
Hee-Youl Kwak$^{2,\dagger}$,
Dae-Young Yun$^{2}$,
and Sang-Hyo Kim$^{1,\dagger}$%
\thanks{$^{\dagger}$Corresponding authors.}\\[0.3em]
$^{1}$Dept. of Electrical and Computer Engineering, Sungkyunkwan University, Suwon, South Korea\\
Email: \{inkikim37@gmail.com, dks1408@g.skku.edu, iamshkim@skku.edu\}\\
$^{2}$Dept. of Electrical, Electronic and Computer Engineering, University of Ulsan, Ulsan, South Korea\\
Email: \{ghy1228@gmail.com, dyyun95@gmail.com\}\\
$^{3}$Dept. of Electrical Engineering, POSTECH, Pohang, South Korea\\
Email: yongjune@postech.ac.kr
}

\maketitle

\begin{abstract}

This paper presents a new protograph-based LDPC code design framework that simultaneously achieves full diversity over block-fading channels (BFCs) and near-capacity performance over additive white Gaussian noise channels. By leveraging a Boolean approximation-based analysis—Diversity Evolution—we derive structural constraints with generalized rootchecks that guarantee full diversity. Building on these constraints, we propose a diversity-aligned protograph template tailored for the two-block BFC ($M=2$) that ensures full diversity under iterative belief propagation decoding. Furthermore, a genetic algorithm guided by density evolution is employed to optimize the protograph edges within this family for improved coding gain. The resulting codes, termed DA-GRP-LDPC codes, simultaneously achieve full diversity and enhanced coding gain, reaching a 0.8~dB gap to the outage limit for the two-block BFC at a block length of 16,896. This demonstrates that the proposed framework effectively bridges the gap between diversity optimality in non-ergodic channels and high coding gain in ergodic channels.
\end{abstract}

\begin{IEEEkeywords}
protograph, block-fading channels, iterative BP decoding, genetic algorithm, LDPC codes
\end{IEEEkeywords}

\section{Introduction}

Low-Density Parity-Check (LDPC) codes~\cite{Gallager62,Richardson01a} are sparse-graph linear block codes that achieve capacity-approaching performance for large block lengths.
Belief-propagation (BP) decoding is particularly well-suited for practical implementation~\cite{Chen05}.
Consequently, LDPC codes have been adopted in a wide range of communication standards such as Digital Video Broadcasting (DVB)-S2/C2/T2, IEEE 802.11 wireless local area network, Advanced Television Systems Committee (ATSC)~3.0, and 3rd Generation Partnership Project (3GPP) 5th Generation-New Radio (5G-NR)~\cite{NR38212}.

Most existing LDPC code designs are optimized for ergodic channel models such as AWGNCs~\cite{Richardson01a,Chung00,Divsalar09}.
This aligns with modern wireless standards, where the codes are typically optimized for AWGNCs and fading effects are addressed through a bit-interleaved coded modulation (BICM) scheme~\cite{NR38211,NR38212}.

In parallel, a substantial body of research has investigated coding techniques for nonergodic channel models, most notably block-fading channels (BFCs)~\cite{Kraidy10,Boutros10,Fang15,FangSurvey15,cKim19,cKim20,Ju22,Ju25,Ahn26}.
In this setting, performance is governed by diversity rather than ergodic capacity~\cite{MalkamakiLeib99,KnoppHum00,Guillen06}.
To address this challenge, root LDPC and root-protograph LDPC codes have been developed to achieve diversity-optimal performance over BFCs~\cite{Boutros10,Fang15,FangSurvey15}.
More recently, generalized-root-protograph (GRP) LDPC codes have extended this framework to higher-rate codes and hybrid automatic repeat request (HARQ) scenarios~\cite{cKim19,cKim20}.

Despite achieving optimal diversity over BFCs, existing root and GRP LDPC codes rely on highly constrained protograph structures, which limit coding gain and lead to suboptimal AWGNC performance.
This raises a fundamental question: Can a single LDPC code simultaneously achieve diversity-optimal performance over BFCs and strong coding gain over AWGNCs and possibly BFCs?
For a BFC with $M=2$ blocks, we propose a protograph-based LDPC design framework that exploits a special graph structure called a generalized rootcheck, thereby achieving diversity optimality while providing enhanced design flexibility for strong coding gain in both AWGNCs and BFCs. We refer to the new codes as diversity-aligned generalized rootcheck protograph LDPC (DA-GRP-LDPC) codes.

The contributions of this work are as follows.
First, we propose a protograph template where full diversity is aligned to all information nodes for the two-block BFC that induces a constrained protograph family and guarantees optimal diversity under iterative BP decoding.
Second, within the proposed template, we employ a genetic algorithm~\cite{Holland75,Elkelesh19} to optimize the protograph for strong coding gain.
As a result, the optimized full-diversity code achieves performance within 0.163~dB of the Shannon limit. Notably, the proposed codes achieve state-of-the-art performance among LDPC codes over the two-block BFC, operating within 0.8~dB of the outage limit at a block length of 16,896.

% As a result, the optimized full diversity code achieves a gap to capacity of 0.163~dB. At a block length of 16896, our optimized code operates within 0.8~dB gap of the outage limit for the two-block BFC. 

\section{Preliminaries}

\subsection{Low-Density Parity-Check Codes}
LDPC codes are linear block codes defined by a sparse parity-check matrix
$\mathbf{H}$ and an equivalent Tanner-graph representation~\cite{Gallager62,Tanner81}.
The Tanner graph consists of variable nodes (VNs) and check nodes (CNs), with an
edge between VN $v_i$ and CN $c_j$ if $H_{ji}=1$. This representation enables
iterative message-passing decoding, commonly referred to as BP decoding~\cite{Chen05}. In practice, low-complexity
approximations of BP, such as min-sum (MS) decoding, are widely used.
Under MS decoding, VN updates sum the channel log-likelihood ratios (LLRs) and incoming CN messages,
whereas CN updates apply sign multiplication and a minimum-magnitude rule to the incoming VN messages~\cite{Chen05}.

\subsection{Protograph Ensembles and QC-LDPC Codes}
In modern LDPC design, quasi-cyclic (QC) LDPC codes are widely used and are constructed by lifting a small bipartite graph, called a protograph~\cite{NR38212}.
Its corresponding matrix representation is referred to as the base matrix.
Let $\mathbf{H}\in\mathbb{F}_2^{(N-K)\times N}$ denote the parity-check matrix of a QC-LDPC code.
In protograph-based design, $\mathbf{H}$ is specified by a base matrix
$\mathbf{H}_p\in\mathbb{F}_2^{(n-k)\times n}$ associated with a protograph
$\mathbb{G}_p(\mathcal{V},\mathcal{C},\mathcal{E})$, where
$\mathcal{V}=\{v_0,\ldots,v_{n-1}\}$ and $\mathcal{C}=\{c_0,\ldots,c_{n-k-1}\}$,
together with lifting information such as cyclic shifts.

Finite-length LDPC codes are constructed by lifting a given protograph with a lifting factor $Z$.
Specifically, each nonzero entry of $\mathbf{H}_p$ is replaced by a $Z\times Z$ circulant permutation matrix, and each zero entry is replaced by a $Z\times Z$ all-zero matrix.
The resulting $\mathbf{H}\in\mathbb{F}_2^{(N-K)\times N}$, where $N=nZ$ and $K=kZ$, exhibits a QC structure and thus defines a QC-LDPC code.
The circulant structure facilitates memory-efficient and hardware-friendly implementations~\cite{Myung05}.

From an ensemble perspective, a QC-LDPC code can be viewed as a finite-length realization of a protograph-based ensemble.
As the lifting factor $Z$ increases, the code performance concentrates around the ensemble average, motivating asymptotic analysis via density evolution (DE)~\cite{Richardson01b,Chung00}.
To evaluate the decoding thresholds with low complexity, we adopt reciprocal channel approximation-based protograph DE (RCA-DE)~\cite{Chung01,Jang22}, and use the resulting threshold as the fitness metric in the proposed genetic algorithm.

\section{Diversity Analysis of LDPC Codes under BP Decoding in Block-Fading Channels}

\subsection{BFC Model and Coded Diversity}
We consider an $M$-block BFC with binary phase-shift keying (BPSK), where $M$
denotes the number of fading blocks~\cite{KnoppHum00,MalkamakiLeib99,Guillen06}. A codeword is
partitioned into $M$ subblocks as $\mathbf{c}=(\mathbf{c}_0,\mathbf{c}_1,\ldots,\mathbf{c}_{M-1})$,
where $\mathbf{c}_m$ is transmitted over the $m$-th fading block. Each subblock
$\mathbf{c}_m$ is modulated to $\mathbf{s}_m$, and the $m$-th block is
modeled as $\mathbf{r}_m = h_m\mathbf{s}_m + \mathbf{n}_m$ for $m=0,\ldots,M-1$,
where $\mathbf{s}_m$ and $\mathbf{r}_m$ denote the transmitted and received
symbol vectors, respectively, $\mathbf{n}_m$ is AWGN, and
$h_m\sim\mathcal{CN}(0,1)$ is constant within the block and independent across
blocks~\cite{KnoppHum00,MalkamakiLeib99}. Although only BPSK is considered in this paper, the results can be readily extended to higher-order modulations.

For transmission over BFCs, achieving diversity across fading blocks is
essential for reliable communication. Since BFCs are nonergodic, performance is
evaluated in terms of the block error rate (BLER) and its high-SNR slope, which reflects
the diversity order~\cite{MalkamakiLeib99}. The achieved
diversity is subject to a Singleton-like bound~\cite{Guillen06}; in particular, for
$M=2$, the maximum rate achieving full diversity is $R=1/2$. This
bound is closely related to the blockwise Hamming distance (i.e., the blockwise
structure of codewords)~\cite{Guillen06}.

In AWGNCs, performance is closely related to Hamming distance and coding gain,
whereas over BFCs it is governed by the blockwise Hamming distance. As a result,
LDPC codes optimized for AWGNCs do not necessarily achieve full diversity over
BFCs, and vice versa~\cite{Ju22,Ju25}. Motivated by this mismatch, we focus on the
$M=2$ case and develop an LDPC design framework that achieves both full diversity
over BFCs and strong performance over AWGNCs.

\begin{figure}[t]
\centering
\subfloat[CN update: AND\label{fig:CN_update}]{%
  \includegraphics[width=0.17\textwidth]{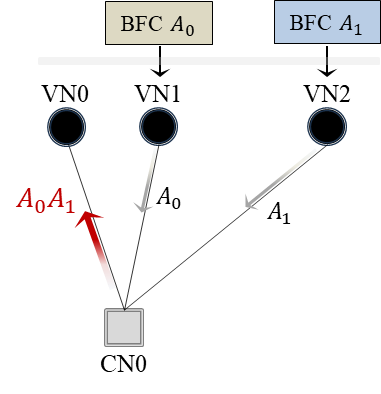}%
}
\hfil
\subfloat[VN update: OR\label{fig:VN_update}]{%
  \includegraphics[width=0.17\textwidth]{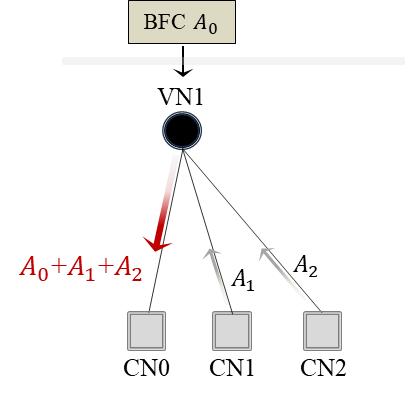}%
}
\caption{Boolean CN/VN updates in DivE (symbolic operations).}
\label{fig:boolean-fading}
\vspace{-0.11in}
\end{figure}

\subsection{Diversity Analysis of LDPC Codes under BP Decoding}
\label{subsec:dive_tool}

We introduce \emph{Diversity Evolution (DivE)}~\cite{Ahn26} to analyze the diversity behavior of protograph-based LDPC codes over BFCs under iterative BP decoding.
We first employ a Boolean approximation $A_m$ of the fading coefficient $h_m$ for each block $m\in\{0,\ldots,M-1\}$ by thresholding the fading/nonfading state, $A_m=\mathbf{1}_{|h_m|^2\gamma \ge \rho_0}$, where $\rho_0$ is a threshold and $\mathbf{1}_E$ denotes the indicator function of an event $E$~\cite{Ju22,Ju25}. 
Let $\mathbf{A}=(A_0,\ldots,A_{M-1})\in\{0,1\}^M$ denote the set of Boolean fading variables, where $A_m$ is the Boolean approximation of the $m$-th block.

DivE proceeds by tracking the evolution of message-level fading behavior through iterative BP decoding via Boolean CN/VN updates (Fig.~\ref{fig:boolean-fading}), as summarized in Algorithm~\ref{alg:DivE}. These updates are symbolic: DivE iteratively updates Boolean algebraic expressions in $\mathbf{A}$,
which we refer to as \emph{Boolean fading functions}. These functions capture the associated diversity behavior~\cite{Ahn26}.

Although implementing Boolean algebra can be tricky for large $M$, DivE can be efficiently derived by running numerical computation of the fading result of MS decoding for all possible individual fading states ($2^M$ cases). This function, referred to as $\mathrm{fadingMSD}$~\cite{Ahn26}, returns a binary vector of length $n$ representing the fading states of the VNs after MS decoding.

Eventually, DivE results in Boolean fading functions (or the truth tables) corresponding to $n$ VNs, respectively.
The diversity order of each VN is determined directly from its final Boolean fading function~\cite{Ju25}. If all information VNs achieve full diversity, then any QC-LDPC code obtained by lifting the protograph can achieve full diversity over BFCs. % 여기에 썼음 정보비트 복호 가능 어쩌고

\begin{algorithm}[t]
\caption{Diversity evolution over a protograph~\cite{Ahn26}}
\label{alg:DivE}
\begin{algorithmic}[1]
\small
\REQUIRE $\mathbb{G}_p$, $\mathbf{A}$, block mapping ${\pi(\cdot)}\in\{0,\ldots,M-1\}^n$. %fading $\mathbf{A}$
\ENSURE Boolean fading functions of VNs
\STATE \textit{// Initialization}
\STATE Assign $A_{\pi(i)}$ to the VN $v_i$ for all $i$.
\STATE \textit{// Main loop}
\FOR{$\ell = 1,\ldots,\ell_{\max}$}
    \STATE Check node update for all CNs as in Fig.~\ref{fig:CN_update}.
    \STATE Variable node update for all VNs as in Fig.~\ref{fig:VN_update}.
\ENDFOR
\STATE \textit{// Output Boolean fading function}
\STATE Output the updated Boolean fading functions of all VNs (addition of all incoming messages).
\end{algorithmic}
\end{algorithm}

\begin{definition}[Block mapping]
A block mapping $\pi:\{0,\ldots,n-1\}\to\{0,\ldots,M-1\}$ assigns each VN
index $i$ (i.e., $v_i$) to a fading block. Under the Boolean fading model, the
channel LLR associated with $v_i$ has fading state $A_{\pi(i)}$.
\end{definition}

\begin{definition}[Boolean fading function]
\label{def:boolean_fading_function}
Under the Boolean approximation, each iterative-decoding message is represented by a Boolean function of the block-fading variables $\{A_0,\ldots,A_{M-1}\}$. The VN-to-CN, CN-to-VN, and a posteriori Boolean functions at iteration $\ell$ are denoted by $F_{v_i\to c_j}^{(\ell)}$, $F_{c_j\to v_i}^{(\ell)}$, and $F_i^{(\ell)}$, respectively. The final Boolean fading function of VN $v_i$ is defined as the fixed point of $F_i^{(\ell)}$, denoted by
$F_i \triangleq \lim_{\ell\to\infty} F_i^{(\ell)}.$
\end{definition} % fixed point 제외?

\begin{definition}[VN diversity]
\label{def:VN_diversity}
For a VN $v_i$, the VN diversity is defined as the diversity order of its final Boolean fading function $F_i$. If $v_i$ is a protograph VN, the same diversity order is assigned to all codeword bits obtained by lifting $v_i$.
\end{definition}

\begin{definition}[Code diversity]
\label{def:code_diversity}
For an LDPC code with information set $\mathcal{I}$, define
$F_c \triangleq \prod_{i \in \mathcal{I}} F_i,$
where $F_i$ denotes the fading function of the $i$-th information VN. The code diversity is defined as the diversity order of $F_c$, where successful BP decoding is identified with recovery of all information VNs.
\end{definition}

\begin{definition}[Full diversity]
\label{def:full_diversity}
A VN $v_i$ is said to achieve full diversity if
$F_i=\sum_{m=0}^{M-1} A_m,$
that is, if its VN diversity is $M$. An LDPC code is said to achieve full diversity if its code diversity is $M$.
\end{definition}

\section{Protograph Design Framework for Block-Fading and AWGN Channels}
The proposed framework for DA-GRP-LDPC code design consists of two stages: (i) defining a protograph template that guarantees full diversity over the BFC with $M=2$, and (ii) selecting protographs from the family induced by the template using a genetic algorithm, where RCA-DE serves as the fitness metric for AWGNC threshold optimization.

\subsection{Generalized Rootcheck Structure}
We briefly review the rootcheck structure of~\cite{Boutros10}, which achieves
full diversity over the BFC after one BP iteration.

\begin{definition}[Rootcheck~\cite{Boutros10}]
A CN $c_j$ is said to be a rootcheck for a VN $v_i$ if $\pi(i) = m$ for some $m \in \{0, \dots, M-1\}$ and, for some $n \neq m$, every $u \in \mathcal{N}_c(j) \setminus \{i\}$ satisfies $\pi(u) = n$.
\end{definition}

\begin{remark}
In this case, a VN $v_i$ connected to a rootcheck achieves full diversity after one BP iteration for $M=2$; such a VN is called a \emph{root node}~\cite{Boutros10}.
\end{remark}

Motivated by this property and the associated structural constraints, we propose
a \emph{generalized rootcheck} concept that provides diversity improvement over
multiple BP iterations. Unlike a rootcheck, it is message-based rather than purely block mapping-based, as illustrated in Fig.~\ref{fig:rootcheck_vs_generalized_rootcheck}.

\begin{definition}[Generalized rootcheck]
\label{def:gen_rootcheck}
A CN $c_j$ is called a generalized rootcheck for a VN
$v_i$ at iteration $\ell$ if $\pi(i)=m$ for
some $m\in\{0,\ldots,M-1\}$ and for some $n\neq m$, every other neighboring VN
$u\in\mathcal{N}_c(j)\setminus\{i\}$ satisfies $F_{v_u\to c_j}^{(\ell-1)}\in\Bigl\{A_n,\ \sum_{n'=0}^{M-1}A_{n'}\Bigr\}$.
In this case, we refer to $v_i$ as the root node of the generalized rootcheck.
\end{definition}

\begin{figure}[t]
\centering
\subfloat[Rootcheck\label{fig:rootcheck}]{%
  \includegraphics[width=0.21\textwidth]{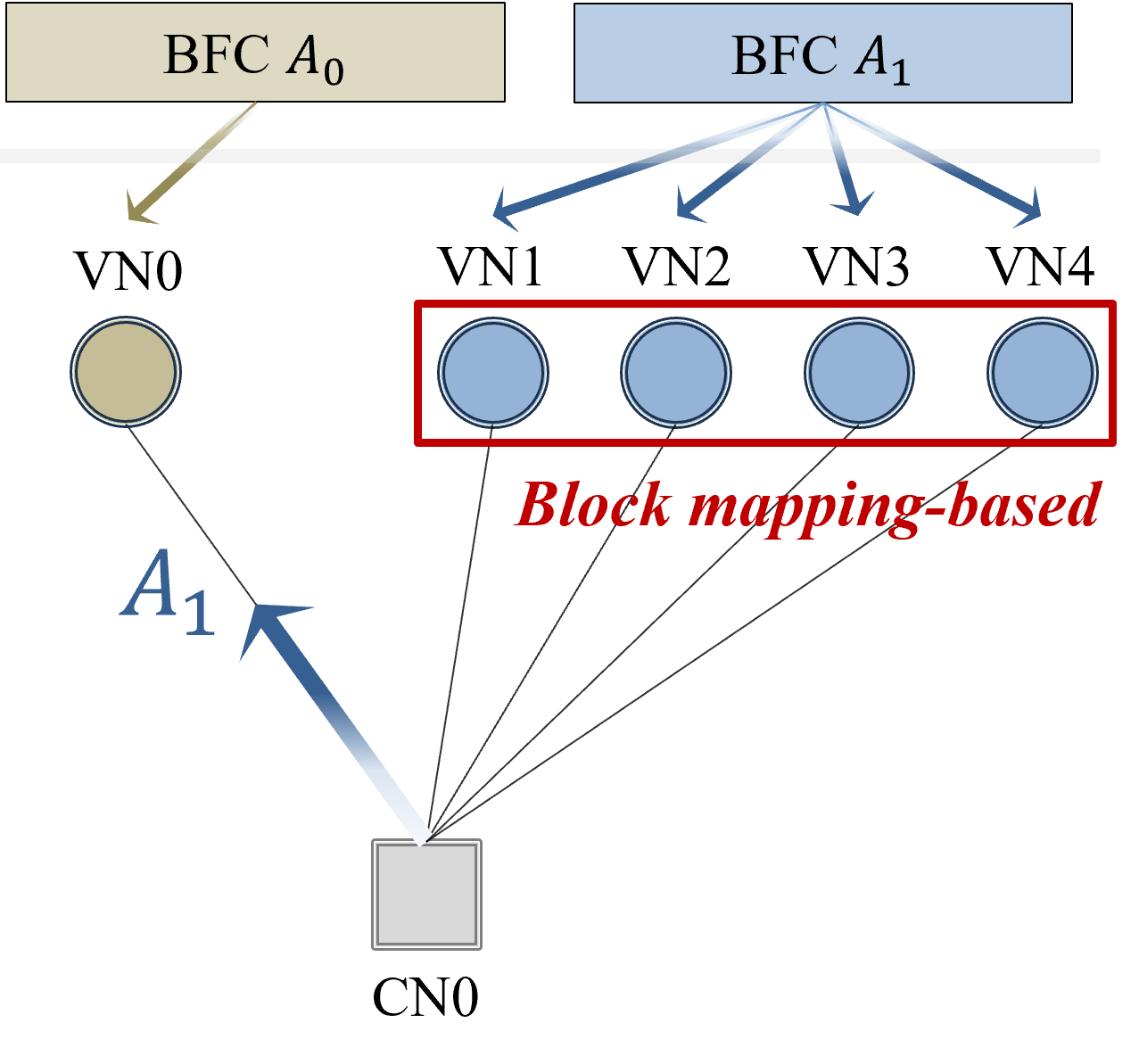}%
}
\hfil
\subfloat[Generalized rootcheck\label{fig:generalized_rootcheck}]{%
  \includegraphics[width=0.21\textwidth]{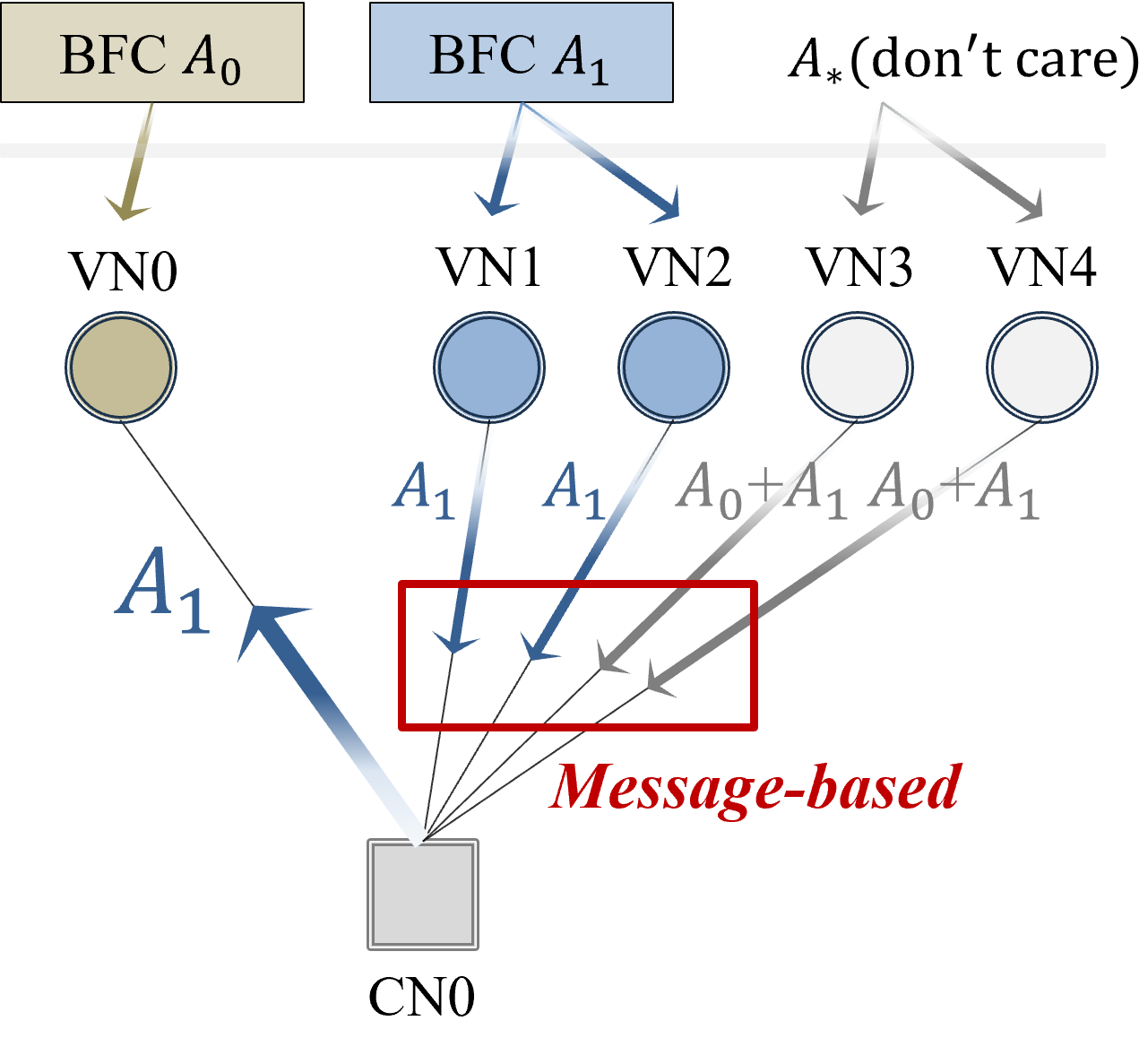}%
}
\caption{Comparison of rootcheck and generalized rootcheck structures.}
\label{fig:rootcheck_vs_generalized_rootcheck}
\vspace{-0.08in}
\end{figure}

The following proposition confirms that the existence of a generalized rootcheck increases the diversity order of the root node across decoding iterations for $M=2$.

\begin{proposition}
\label{prop:gen_rootcheck_div_increase}
For the BFC with $M=2$, if $v_i$ is connected to a generalized rootcheck at iteration $\ell$ and is the root node, then $v_i$ attains full diversity after iteration $\ell$.
\end{proposition}

\begin{proof}
Assume $v_i$ is assigned to block~$0$, and let $c_j$ be a
generalized rootcheck adjacent to $v_i$ at iteration $\ell$ for $M=2$.
Under DivE, the CN update (AND) combines the incoming functions
from $\mathcal{N}_c(j)\setminus\{i\}$. By the generalized rootcheck condition,
these include the opposite fading variable $A_1$ (and possibly $A_0+A_1$, i.e.,
a full diversity expression). Since $A_1\cdot A_1=A_1$ and
$(A_0+A_1)\cdot A_1=A_1$ in Boolean algebra, we obtain
$F_{c_j\to v_i}^{(\ell)}=A_1$. The subsequent VN update (OR) yields
$F_i^{(\ell)}=A_0+A_1$, i.e., $v_i$ attains full diversity.
\end{proof}

For $M=2$, the above result implies a sufficient condition for full diversity.

\begin{proposition}
\label{prop:all_gen_rootcheck_full_M=2}
For the BFC with $M=2$, if every information VN becomes a root node at some decoding iteration, then the LDPC code achieves full diversity.
\end{proposition}

\begin{proof}
By Proposition~\ref{prop:gen_rootcheck_div_increase}, each information VN achieves full diversity. Thus, the code achieves full diversity by Definition~\ref{def:full_diversity}.
\end{proof}

\subsection{Protograph Family Achieving Full Diversity} % Protograph family?

Using the generalized rootchecks, we introduce a diversity-aligned (DA) \emph{protograph template} for constructing protographs that achieve full diversity over the BFC with $M=2$. Fig.~\ref{fig:protograph_template} illustrates the proposed template for $(n,k)=(32,16)$.

\begin{figure}[t]
\centering
\includegraphics[width=0.80\linewidth]{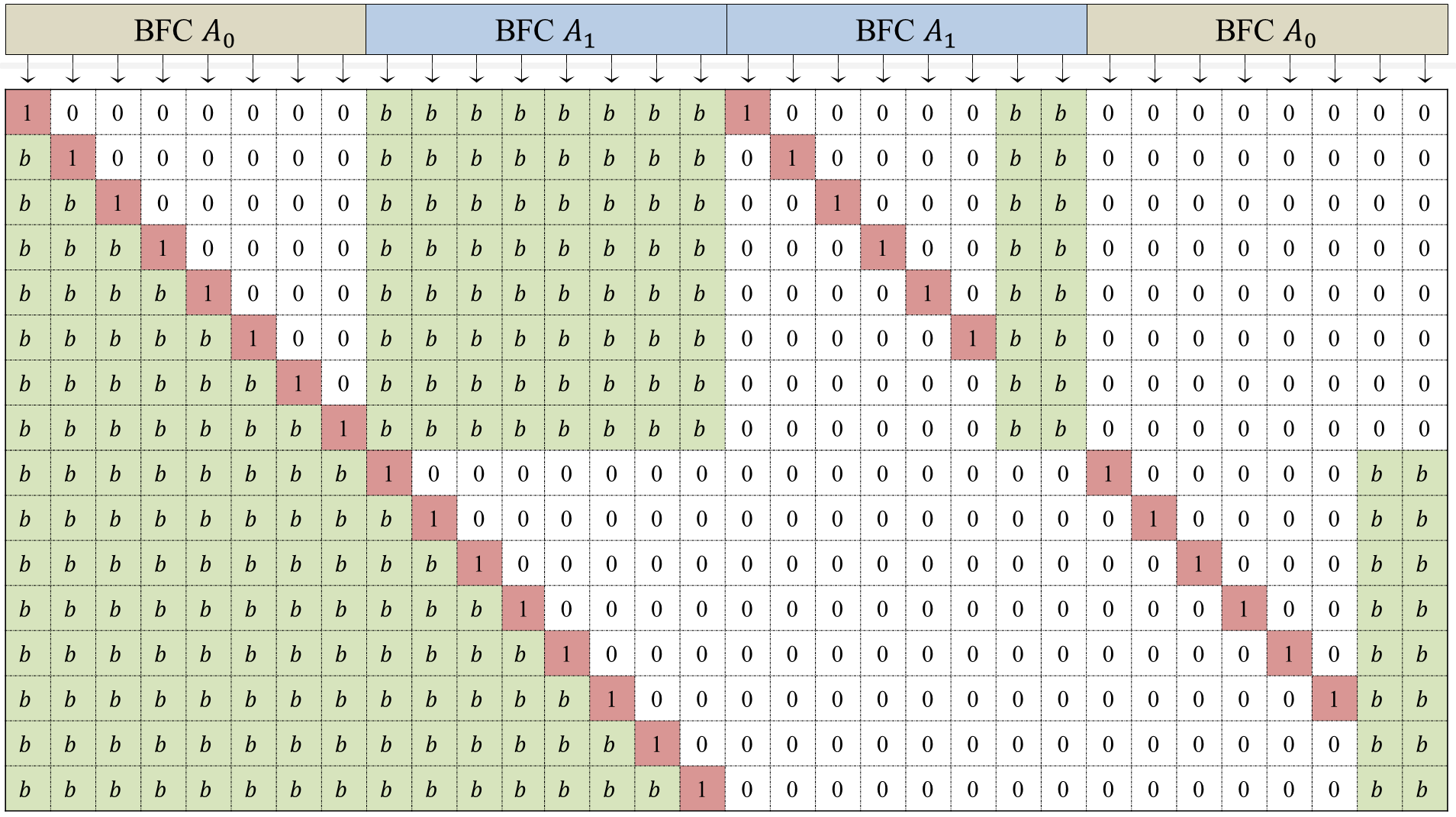}
\caption{Proposed protograph template for the two-block BFC with $n=32$ and $k=16$, achieving full diversity. Entries labeled $b\in\{0,1\}$ are designable.}
\label{fig:protograph_template}
\vspace{-0.09in}
\end{figure}

\begin{definition}[Protograph template]
\label{def:pd_template}
A protograph template is specified by $\mathcal{T}$ and induces
a constrained \emph{protograph family} $\mathcal{H}(\mathcal{T})$. For each
$\mathbf{H}_p=[h_{ji}]\in\mathcal{H}(\mathcal{T})$, the template fixes a subset
of entries of $\mathbf{H}_p$ to $0$ or $1$ and specifies the designable index set
$\mathcal{D}$, where $(j,i)\in\mathcal{D}$ implies $h_{ji}=b$ with $b\in\{0,1\}$.
We focus on the rate-$1/2$ case ($k=n/2$), the maximum rate achieving full diversity for the BFC with $M=2$ under the Singleton-like bound. Here, $v_0,\ldots,v_{n/2-1}$ are information VNs and $v_{n/2},\ldots,v_{n-1}$ are parity VNs.
\end{definition}

\begin{theorem}
\label{thm:protograph_template_proof}
Consider the protograph template $\mathcal{T}$ of Definition~\ref{def:pd_template} and the induced protograph family
$\mathcal{H}(\mathcal{T})$. Suppose that the following conditions hold:
\begin{itemize}[label=-, labelwidth=1em, labelsep=0.5em, leftmargin=!]
    \item \textbf{Block mapping:} %\boldmath$(\pi_{\mathcal T})$:}
    $v_i$ is assigned to block~$0$ for
    $i\in[0,n/4)\cup[3n/4,n)$, and to block~$1$ for $i\in[n/4,3n/4)$.

    \item \textbf{Weight constraints:} 
    $\sum_{i:(j,i)\in\mathcal{D}} h_{ji} \ge 1$ for each $j \in \{n/4-2,n/4-1,n/2-2,n/2-1\}$, and
    $\sum_{j} h_{ji} \ge 1$ for each $i \in \{3n/4-2, 3n/4-1, n-2, n-1\}$.
    
\end{itemize}

Under these conditions, any protograph $\mathbf{H}_p\in\mathcal{H}(\mathcal{T})$
achieves full diversity over the two-block BFC under iterative BP decoding by iteration $n/4$.
\end{theorem}

\begin{proof}
It is guaranteed that $v_{\ell-1}$ and $v_{n/4+\ell-1}$ achieve full diversity by decoding iteration $\ell$ due to the recursive construction of generalized rootchecks in the template. This is shown via DivE by tracking the generation of generalized rootchecks across iterations.

\begin{enumerate}[label=\arabic*)]
\item \textit{Starting with a pair of rootchecks ($\ell=1$):}
By the fixed entries imposed by the protograph template, $v_0$ and $v_{n/4}$ are the
root nodes of rootchecks $c_0$ and $c_{n/4}$, respectively. Hence, they are guaranteed to achieve full diversity after the first BP iteration.

\item \textit{Propagation of full diversity via the generalized rootcheck structure ($\ell\ge 2$):}
For $\ell=2,\ldots,n/4$, the protograph template constraints guarantee that
$v_{\ell-1}$ and $v_{n/4+\ell-1}$ are root nodes of the generalized rootchecks
$c_{\ell-1}$ and $c_{n/4+\ell-1}$, respectively. The other VNs neighboring these CNs depend either on the opposite fading variable alone or on full diversity expressions obtained in previous iterations.
Therefore, by Proposition~\ref{prop:gen_rootcheck_div_increase}, these VNs achieve full diversity by iteration $\ell$.

\item \textit{Full diversity guarantee:}
Since two information VNs are guaranteed to achieve full diversity at each iteration, all information VNs achieve full diversity by iteration $n/4$. This stepwise guarantee follows from the lower triangular structure of the protograph template. Hence, by Proposition~\ref{prop:all_gen_rootcheck_full_M=2}, any $\mathbf{H}_p\in\mathcal{H}(\mathcal{T})$ achieves full diversity under iterative BP decoding.
\end{enumerate}

\end{proof}

Theorem~\ref{thm:protograph_template_proof} guarantees full diversity over the two-block BFC for all protographs in $\mathcal{H}(\mathcal{T})$.
Within this constrained family, we can search for a protograph with improved coding gain.

\begin{remark}
The degree-1 VNs in the proposed template improve finite-length performance~\cite{Divsalar09}, reduce the search space, and facilitate the generalized rootcheck condition.
\end{remark}

\subsection{Genetic Algorithm-Based Protograph Optimization}

\begin{figure}[t]
\centering
\includegraphics[width=0.75\linewidth]{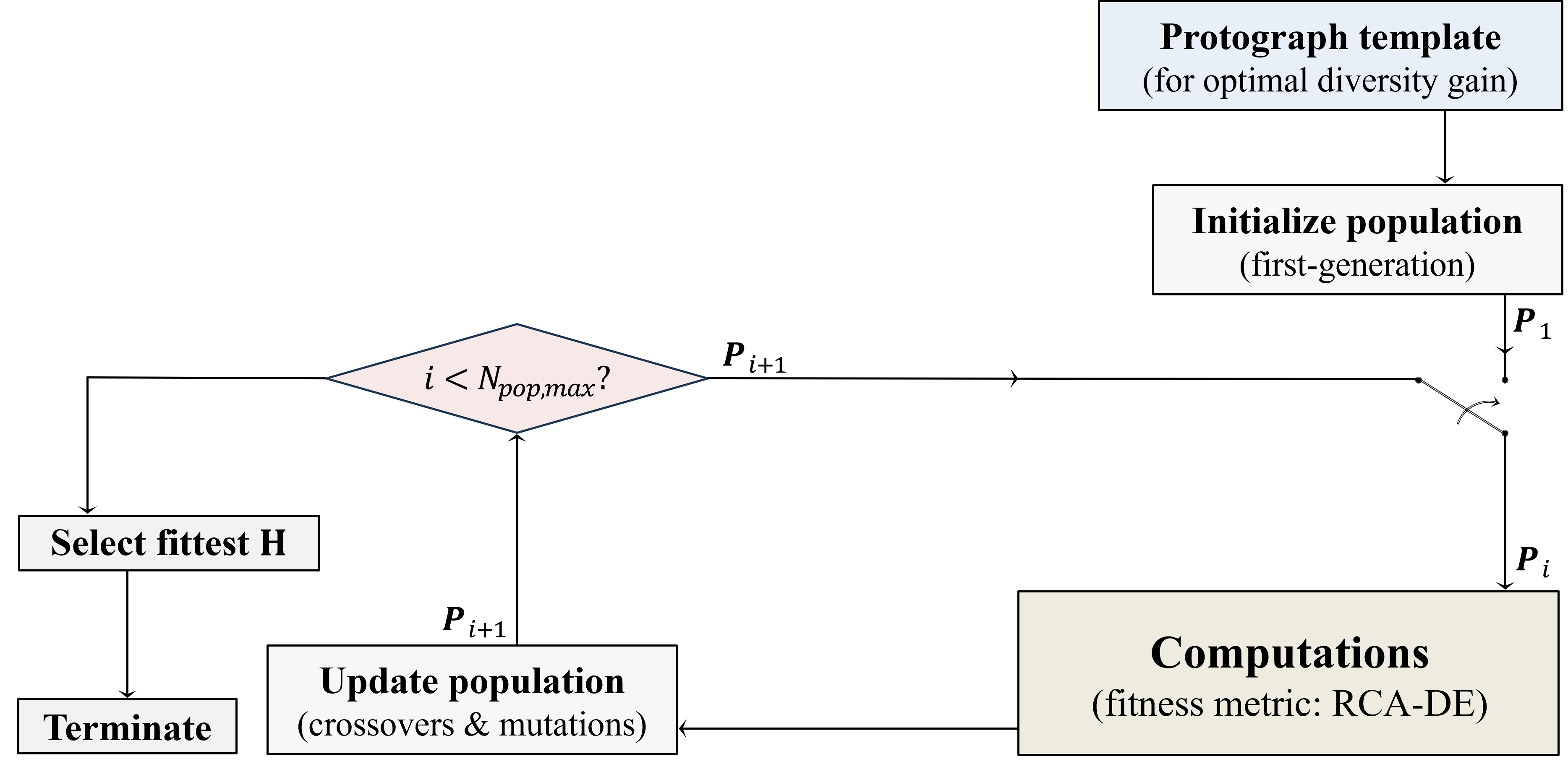}
\caption{Flowchart of GA-based protograph optimization for AWGNCs.}
\vspace{-0.09in}
\label{fig:ga_flow}

\end{figure}

%We pursue to search for a good protograph instance among the template-induced protograph family. 
We seek to find a good protograph instance among the template-induced protograph family.
The search over the protograph family is a constrained binary combinatorial optimization problem.
Because the RCA-DE threshold defines a highly nonconvex and derivative-inaccessible objective over this discrete space, we employ a genetic algorithm (GA)~\cite{Holland75,Elkelesh19} for efficient global exploration.
Fig.~\ref{fig:ga_flow} summarizes the proposed GA-based optimization procedure.

\begin{enumerate}[label=\arabic*)]
\item \textit{Population initialization:}\\
The designable entries $(j,i)\in\mathcal{D}$ are randomly initialized with
$b\in\{0,1\}$ to generate an initial population of a predefined size.

\item \textit{Fitness evaluation (via RCA-DE):}\\
For each candidate $\mathbf{H}_p\in\mathcal{H}(\mathcal{T})$ in generation $G_i$,
RCA-DE is performed, and the resulting decoding threshold is used as the fitness
metric for AWGNC performance. The population is then ranked by fitness.

\item \textit{Population update:}\\
The next generation $G_{i+1}$ is formed through elite retention, crossover, and
mutation. The top $T$ candidates in $G_i$ (set to $T=5$ in our search) are retained and used as parents.
Crossover produces $8\binom{T}{2}$ offspring in total: $6\binom{T}{2}$ from VN-based crossover and $2\binom{T}{2}$ from CN-based crossover. In the former, VNs are partitioned into four groups by block mapping and information/parity type; in the latter, CNs are partitioned according to whether they act as generalized rootchecks for block~0 or block~1 VNs.
Mutation is applied only to $\mathcal{D}$ by adding, deleting, or swapping a single designable entry $h_{ji}$. 

\item \textit{Validity check:}\\
Each candidate is examined to ensure that the parity submatrix of $\mathbf{H}_p$ is full rank to enable systematic code construction\footnote{The results in this version were obtained after enforcing this condition.} and that every information VN is connected to multiple CNs not adjacent to any degree-1 VN. Candidates that do not satisfy these conditions are discarded.

\end{enumerate}

\noindent

The GA parameters were chosen empirically, as further increases yielded only marginal performance improvement at substantially higher computational cost.
For finite-length code construction, the optimized protograph is lifted using the progressive edge-growth (PEG) technique~\cite{Hu05}. 

\begin{table}[t]
\caption{Protograph thresholds and gap to capacity over BI-AWGNCs.}
\label{tab:rcade_gap_thr}
\centering
\renewcommand{\arraystretch}{1.1}
\setlength{\tabcolsep}{6pt}
\begin{tabular*}{0.8\columnwidth}{@{\hspace{4pt}\extracolsep{\fill}}lcccc@{\hspace{4pt}}}
\toprule
$R=1/2$ & Proposed & 5G-NR & IRP2 & RP \\
\midrule
$\gamma_{\mathrm{th}}$ [dB]  & 0.350 & 0.440 & 0.730 & 1.102 \\
$\Delta_{\mathrm{cap}}$ [dB] & 0.163 & 0.253 & 0.543 & 0.915 \\
\bottomrule
\end{tabular*}
\vspace{-0.09in}
\end{table}

\section{Numerical Results}

In this section, we evaluate the proposed DA-GRP-LDPC codes over the BFC with $M=2$ and the AWGNC, using root-protograph (RP) LDPC codes~\cite{Boutros10}, improved root-protograph 2 (IRP2) LDPC codes~\cite{Fang15}, and 5G-NR LDPC codes~\cite{NR38212} as benchmarks. The IRP2 codes serve as advanced RP baselines with enhanced coding gain.

Fig.~\ref{fig:performance_BFC} shows the BLER performance over the two-block BFC for $(N,R)=(16896,1/2)$. The proposed code, RP LDPC code, and IRP2 LDPC code exhibit full diversity behavior with comparable high-SNR slopes, whereas the 5G-NR baseline with arbitrary block mapping shows a shallower slope due to diversity loss. Beyond achieving full diversity, the proposed code outperforms both the RP and IRP2 baselines, indicating additional coding gain. Consequently, it achieves state-of-the-art BLER performance among LDPC codes over the two-block BFC, operating within 0.8~dB of the outage limit at a block length of 16,896.

Fig.~\ref{fig:performance_AWGNC} compares the bit error rate (BER) and BLER performance over the binary-input AWGN channel (BI-AWGNC). The code lengths considered in Fig.~\ref{fig:performance_AWGNC_1} and Fig.~\ref{fig:performance_AWGNC_2} are $N=7744$ and $N=16896$, respectively, with a code rate of $R=1/2$. At both lengths, the proposed DA-GRP-LDPC codes outperform all baselines, demonstrating that the RCA-DE-guided GA search yields additional coding gain within the template design space. 
This is supported by Table~\ref{tab:rcade_gap_thr}, where the proposed protograph exhibits the smallest gap to capacity $\Delta_{\mathrm{cap}}$ among the listed BI-AWGNC thresholds $\gamma_{\mathrm{th}}$.
Although the GA search can find protographs with thresholds as low as 0.301~dB, we select one with a threshold of 0.350~dB to avoid an error floor in finite-length performance while preserving the full diversity guarantee over the BFC with $M=2$.

\begin{figure}[t]
\centering
\includegraphics[width=0.55\linewidth]{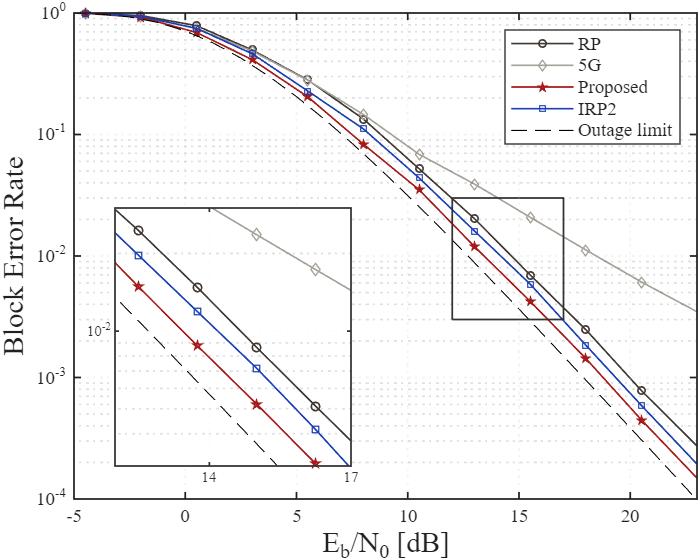}
\caption{BLER performance over the two-block BFC for the 5G-NR, RP, IRP2, and proposed DA-GRP-LDPC codes with $(N,R)=(16896,1/2)$. The outage limit represents the information-theoretic BLER lower bound for the two-block BFC under BPSK modulation. PEG-based lifting is applied with $Z=1056$ for the RP code, $Z=528$ for the IRP2 code, and $Z=384$ for the 5G-NR and proposed codes.}
\label{fig:performance_BFC}
\vspace{-0.08in}
\end{figure}

\begin{figure}[t]
\centering
\subfloat[$N=7744$, $R=1/2$\label{fig:performance_AWGNC_1}]{%
  \includegraphics[width=0.22\textwidth]{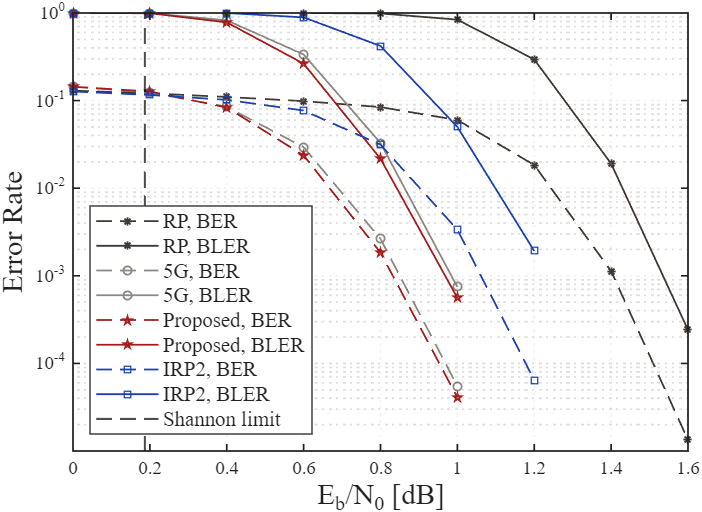}%
}
\hfil
\subfloat[$N=16896$, $R=1/2$\label{fig:performance_AWGNC_2}]{%
  \includegraphics[width=0.22\textwidth]{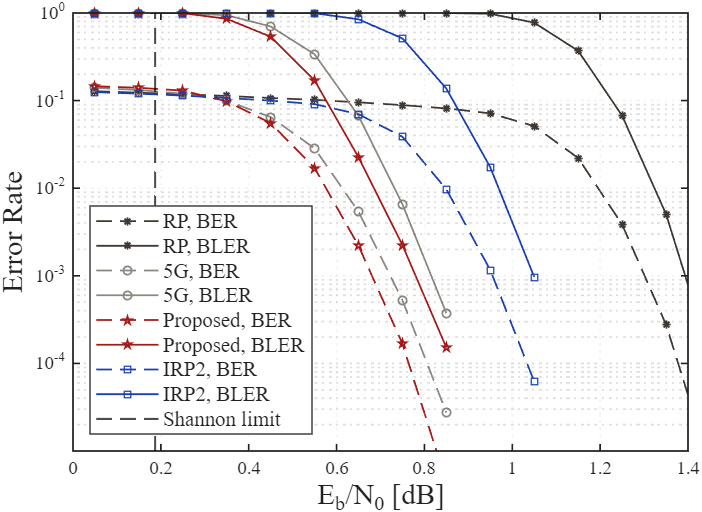}%
}
\caption{BER and BLER performance over the BI-AWGNC for the 5G-NR, RP, IRP2, and proposed DA-GRP-LDPC codes at $R=1/2$: (a) $N=7744$; (b) $N=16896$.} 
%The Shannon limit denotes the theoretical bound.}

\label{fig:performance_AWGNC}
\vspace{-0.09in}
\end{figure}

\section{Conclusion}
We proposed a protograph-based LDPC code design framework that achieves full diversity
over the BFC with $M=2$ while maintaining strong coding gain in both the AWGNC and the BFC.
Full diversity is guaranteed for our DA protograph family under iterative decoding. 
Coding gain is then optimized within the DA protograph family using an RCA-DE-guided 
GA targeting low protograph thresholds for AWGNC performance.  
Numerical results confirm that the resulting LDPC codes exhibit full diversity over the two-block BFC and achieve better BLER performance over AWGNCs across multiple code lengths
compared with the considered baselines. 

The proposed design inherently offers scalability to $M>2$, which is currently being investigated as ongoing work. Furthermore, future directions include exploring applications to modified codes and HARQ.

\end{document}